\documentclass[11pt]{article}

\usepackage{graphicx,amssymb,amsmath}
\usepackage{epsfig}
\usepackage{epsf}
\usepackage{datetime}
\usepackage{latexsym}

\usepackage[small]{caption}
\usepackage{amsfonts}
\usepackage{fancyvrb}
\usepackage{fancyhdr}
\usepackage{amssymb}
\usepackage{amsthm}
\usepackage{graphics}
\usepackage{psfrag}
\usepackage{bbm}
\usepackage{hyperref}

\newtheorem{theorem}{Theorem}
\newtheorem{lemma}{Lemma}

\newtheorem{corollary}{Corollary}

\def\etal{{et~al.}}
\def\eg{{e.g.}}
\def\ie{{i.e.}}

\newcommand{\opt}{\textsf{OPT}}

\newcommand{\conv}{{\rm conv}}
\newcommand{\len}{{\rm len}}
\newcommand{\per}{{\rm per}}
\newcommand{\dist}{{\rm dist}}
\newcommand{\diam}{{\rm diam}}

\newcommand{\RR}{\mathbb{R}} %  set of real numbers
\newcommand{\eps}{\varepsilon}

\def\G{\mathcal G}
\def\H{\mathcal H}
\def\R{\mathcal R}
\def\Q{\mathcal Q}

\newcommand{\old}[1]{{}}
\newcommand{\later}[1]{{}}

\title{\textsc{Constant-Factor Approximation for TSP with Disks}}

\author{Adrian Dumitrescu\thanks{Department of Computer Science,
University of Wisconsin--Milwaukee, WI, USA\@.
Email:~\texttt{dumitres@uwm.edu}.}
\qquad
Csaba D. T\'oth\thanks{Department of Mathematics, California State
  University, Northridge, Los Angeles, CA;
and Department of Computer Science, Tufts University, Medford, MA, USA\@.
Email:~\texttt{cdtoth@acm.org}.
Research by this author was supported in part by the NSF award CCF-1423615.}
}

\begin{document}

\maketitle

\begin{abstract}
We revisit the traveling salesman problem with neighborhoods
(TSPN) and present the first constant-ratio approximation for disks in
the plane: Given a set of $n$ disks in the plane, a TSP tour whose
length is at most $O(1)$ times the optimal can be computed in time that
is polynomial in $n$. Our result is the first constant-ratio approximation
for a class of planar convex bodies of arbitrary size and arbitrary intersections.

In order to achieve a $O(1)$-approximation, we reduce the traveling
salesman problem with disks, up to constant factors, to a minimum
weight hitting set problem in a geometric hypergraph.
The connection between TSPN and hitting sets in geometric hypergraphs,
established here, is likely to have future applications.

\end{abstract}

\section{Introduction} \label{sec:intro}

In the Euclidean Traveling Salesman Problem (ETSP), given a set
of points in the Euclidean space $\RR^d$, $d \geq 2$,
one seeks a shortest closed curve (a.k.a. \emph{tour}) that visits each point.
In the TSP \emph{with neighborhoods} (TSPN),  each point is replaced by a point-set,
called \emph{region} or \emph{neighborhood}, and the TSP tour must visit at least one
point in each region, \ie, it must intersect each region.
The oldest record that we could trace of this variant goes back to
Arkin and Hassin~\cite{AH94}. Since the Euclidean TSP is known to be NP-hard in $\RR^d$
for every $d \geq 2$ \cite{GGJ76,GJ79,Pa77}, TSPN is also NP-hard
for every $d \geq 2$. TSP is recognized as one of the corner-stone problems
in combinatorial optimization. Other related problems in geometric network optimization
can be found in the two (somewhat outdated) surveys by Mitchell~\cite{Mi00,Mi04}.

It is known that the Euclidean TSP admits a polynomial-time
approximation scheme in $\RR^d$, where $d=O(1)$, due to
classic results of Arora~\cite{Ar98} and Mitchell~\cite{Mi99}.
Subsequent running time improvements have been obtained by
Rao and Smith~\cite{RS98}; specifically, the running time
of their PTAS is $O(f(\eps) \, n \log{n})$, where $f(\eps)$ grows
exponentially in $1/\eps$. In contrast, TSPN is generally harder to approximate.
Typically, somewhat better approximations are available when the neighborhoods
are \emph{pairwise disjoint}, or \emph{fat}, or have \emph{comparable sizes}.
We briefly review some of the previous work concerning approximation
algorithms for TSPN.

\paragraph{Related work.}
Arkin and Hassin~\cite{AH94} gave constant-factor approximations for
translates of a connected region, and more generally, for neighborhoods
of pairwise \emph{parallel} diameters, where the ratio between the
longest and the shortest diameter is bounded by a constant.
Dumitrescu and Mitchell~\cite{DM03} extended the above result to
connected neighborhoods with comparable diameters.
Bodlaender~\etal~\cite{BFG+09} described a PTAS for TSPN with disjoint fat
neighborhoods of about the same size in $\RR^d$, where $d$ is constant
(this includes the case of disjoint unit disks in the plane).
Earlier Dumitrescu and Mitchell~\cite{DM03} proposed a PTAS for TSPN with
fat neighborhoods of about the same size and bounded depth in the plane,
where Spirkl~\cite{Sp14} recently reported and filled a gap; see also a follow-up
note in~\cite{Mi16a}.

Mata and Mitchell~\cite{MM95} gave a $O(\log{n})$-approximation for TSPN with $n$
connected and arbitrarily intersecting neighborhoods in the plane; see also~\cite{BE97}.
Elbassioni~\etal~\cite{EFS09} and Gudmundsson and Levcopoulos~\cite{GL99}
improved the running time of the algorithm.
The $O(\log{n})$-approximation relies on the following early result by
Levcopoulos and Lingas~\cite{LL84}:
Every (simple) rectilinear polygon $P$ with $n$ vertices, $r$ of which are reflex,
can be partitioned in $O(n \log{n})$ time into rectangles whose total perimeter
is $\log{r}$ times the perimeter of $P$.

Using an approximation algorithm due to Slavik~\cite{Sl97} for Euclidean group
TSP, de Berg~\etal~\cite{BGK+05} obtained constant-factor approximations
for disjoint fat convex regions (of arbitrary diameters) in the plane.
Subsequently, Elbassioni~\etal~\cite{EFS09} gave constant-factor
approximations for arbitrarily intersecting fat convex regions of comparable size.
Preliminary work by Mitchell~\cite{Mi07} gave a PTAS for planar regions of bounded depth
and arbitrary size, in particular for disjoint fat regions. Chan and Jiang~\cite{CJ14}
gave a PTAS for fat, weakly disjoint regions in metric spaces of constant doubling
dimension (combining an earlier QPTAS by Chan and Elbassioni~\cite{CE11} with
a PTAS for TSP in doubling metrics by Bartal~\etal~\cite{BGK12}).

Disks and balls are undoubtedly among the simplest neighborhood types~\cite{AH94,DM03,KS13}.
TSPN for disks is NP-hard, and it remains so for congruent disks, since when
the disk centers are fixed and the radius tends to zero, the problem reduces to TSP for points.
Regarding approximations, the case of congruent balls is relatively well understood:
Given a set of $n$ congruent (say, unit) balls in $\RR^d$, a TSP tour
whose length is at most $O(1)$ times the optimal can be computed in
polynomial time, when $d$ is constant~\cite{DT16}. However, for disks of
arbitrary radii and intersections, no constant-ratio approximation was known.
Some of the difficulties with disks of arbitrary radii in the plane
where uncovered in~\cite{DT15}.

Recent work of Dumitrescu and T\'oth~\cite{DT16} focused on unbounded neighborhoods,
such as lines or hyperplanes: They gave a constant-factor approximation for TSPN
with $n$ hyperplanes in $\RR^d$ in $O(n)$ time;
and a $O(\log^3 n)$-approximation for $n$ lines in $\RR^d$ in time polynomial in $n$,
where $d$ is constant.
In contrast, the current paper considers TSPN with arbitrary disks in $\RR^2$,
which requires quite different approximation techniques and new ideas.

\paragraph{Degree of approximation.}
Regarding the degree of approximation achievable, TSPN with arbitrary
neighborhoods is generally APX-hard~\cite{BGK+05,DO08,SS05}, and it remains
so even for segments of nearly the same length~\cite{EFS09}.
For disconnected neighborhoods, TSPN cannot be approximated within
any constant ratio unless $P=NP$~\cite{SS05}.
%such as point-pair neighborhoods~\cite{DO08}.
Further, approximating TSPN for (arbitrary) connected neighborhoods
in the plane within a factor smaller than 2 is NP-hard~\cite{SS05}.
Delineating the class of neighborhoods for which constant-factor
approximations are possible remains mysterious, at least at the moment.
It is conjectured that approximating TSPN for disconnected regions in
the plane within a $O(\log^{1/2} n)$ factor is intractable unless $P=NP$~\cite{SS05}.
Similarly, it is conjectured that approximating TSPN for connected
regions in $\RR^3$ within a $O(\log^{1/2} n)$ factor and for disconnected
regions in $\RR^3$ within a $O(\log^{2/3} n)$ factor~\cite{SS05} are probably intractable.

\paragraph{Our results.}
In this paper we present a polynomial-time (deterministic) algorithm that,
given a set of $n$ disks in $\RR^2$ (with arbitrary radii and intersections),
returns a TSP tour whose length is $O(1)$ times the optimal.

\begin{theorem}\label{thm:disk}
Given a set of $n$ disks in the plane, a TSP tour whose length is at
most $O(1)$ times the optimal can be computed in time polynomial in $n$.
\end{theorem}

In their seminal paper on TSPN, Arkin and Hassin~\cite{AH94}
suggested disks as the most natural type of neighborhood---in
which the traveling salesman can meet each potential buyer
at some point close to the respective buyer's location.
Here the radius of each disk indicates how much each potential
buyer is willing to travel to meet the salesman.
While constant-ratio approximations for disks of the same (or comparable)
radius~\cite{AH94,DM03} and for disjoint disks of arbitrary radii~\cite{BGK+05}
have been obtained early on, the case of disks with arbitrary radii
and arbitrary intersections has remained open until now.

\section{Preliminaries}

We achieve a $O(1)$-approximation for TSP with disks by reducing the
problem, up to constant factors, to a minimum weight hitting set
problem in a geometric hypergraph, for which a constant-factor approximation algorithm
was found only recently~\cite{CGKS12}.

\paragraph{Hitting sets.} Hitting sets are defined in general in terms
of hypergraphs (i.e., set systems or range spaces).
A hypergraph is a pair $\G=(V,E)$ where $V$ is a finite vertex set and
$E \subset 2^V$ is a finite collection of subsets of $V$ (called edges).
In a geometric (primal) hypergraph, the vertex set $V$ is a finite set
of $n$ points in Euclidean space $\RR^d$, and all sets in $E$ are of the
form $V\cap Q$ where $Q$ is a certain geometric shape of bounded
description complexity, e.g., halfspace, ball, triangle, axis-aligned
rectangle, etc. Geometric hypergraphs often have nice properties
such as bounded VC-dimension or bounded union complexity; see~\cite{ERS05,PA95}.

A \emph{hitting set} in a hypergraph $\G=(V,E)$ is a subset of
vertices $H\subseteq V$ such that every hyperedge in $E$ contains some
point in $H$. The \emph{minimum hitting set} (MHS) problem asks for a
hitting set of minimum cardinality in a given hypergraph. The
\emph{minimum weight hitting set} (MWHS) problem asks for a hitting
set of minimum weight in a given hypergraph with vertex weights
$w~:~V\rightarrow \RR^+$.

Br\"onnimann and Goodrich~\cite{BG95} gave a $O(\log \opt)$-approximation
for MHS in geometric hypergraphs using LP-relaxations
and the fact that geometric hypergraphs have bounded VC-dimension.
Clarkson and Varadarajan~\cite{CV07} gave a $O(\log \log n)$-approximation for
\emph{some} geometric hypergraphs, by observing a connection between
hitting sets and the combinatorial complexity of the union of the
corresponding geometric objects. Mustafa and Ray~\cite{MR10} gave a
PTAS for MHS with disks and pseudo-disks in the plane using a
local search paradigm; see also~\cite{AES12,BGM+15}. However, this
method does not seem to extend to the weighted version (MWHS).

Varadarajan~\cite{Va09} gave a $O(\log \log n)$-approximation for MWHS,
extending the results from~\cite{CV07}. His approach was further
extended by Chan~\etal~\cite{CGKS12} who obtained a randomized
polynomial-time $O(1)$-approximation algorithm for MWHS in geometric
hypergraphs of linear union complexity, including geometric hypergraphs
defined by disks in $\RR^2$~\cite{KLPS86}; their algorithm can be
derandomized~\cite[Section 3]{CGKS12}. Specifically, we use the following
result due to Chan~\etal~\cite{CGKS12}.

\begin{theorem} {\rm \cite[Corollary~1.4 and Section 3]{CGKS12}} \label{thm:chan}
There is a polynomial-time (deterministic) $O(1)$-approximation algorithm
for the minimum weight hitting set problem for disks in $\RR^2$.
\end{theorem}

\paragraph{Definitions.}
Let $\R$ be a set of regions (neighborhoods) in $\RR^2$.
An optimal TSP tour for $\R$, denoted by $\opt(\R)$, is a shortest
closed curve in the plane that intersects every region in~$\R$;
when $\R$ is clear from the context, $\opt(\R)$ and $\opt$ are used
interchangeably.

The Euclidean length of a curve $\gamma$ is denoted by $\len(\gamma)$.
Similarly, the total (Euclidean) length of the edges
of a geometric graph $G$ is denoted by $\len(G)$. The perimeter of a
polygon $P$ is denoted by $\per(P)$; the boundary and the interior
of a region $R$ are denoted by $\partial R$ and $R^{\circ}$, respectively;
the convex hull of a planar set $S$ is denoted by $\conv(S)$.

The distance between two planar point sets $S_1,S_2\subset \RR^2$, is
$\dist(S_1,S_2)=\inf\{\dist(s_1,s_2): s_1\in S_1, s_2\in S_2\}$.
The distance between a point set $S_1$ and a geometric graph $G$ is
defined as $\dist(S_1,G):=\dist(S_1,S_G)$,
where $S_G$ is the set of all points at vertices and on the edges of $G$.

\paragraph{Algorithm Outline.}
Given a set $S$ of $n$ disks in Euclidean plane, we construct a
connected geometric graph $G$ that intersects every disk
in $S$ and such that $\len(G) = O(\len(\opt))$. An Eulerian tour of
the multi-graph obtained by doubling each edge of $G$ visits each disk
and its length is $2 \, \len(G) = O(\len(\opt))$, as desired.

The graph $G$ is the union of three geometric graphs,
$G_1$, $G_2$ and $G_3$. The graph $G_1$ is a $O(1)$-approximation of
an optimal tour for a maximal subset of pairwise disjoint disks in $S$;
this step is based on earlier results~\cite{BGK+05,EFS09} (Section~\ref{sec:pre}).
The graph $G_2$ connects $G_1$ to nearby disks that are guaranteed to be at distance
at most $\len(\opt)/n$ from $G_1$. The graph $G_3$ connects any remaining disks to $G_1$;
this step is based on recent results on minimum weight hitting sets due to
Chan~\etal~\cite{CGKS12} (Section~\ref{sec:hitting}).

The interface between TSP and the hitting set problem is established by
a quadtree subdivision~\cite[Ch.~14]{BCK+08}.
Previously, Arora~\cite{Ar98} and Mitchell~\cite{Mi10} used quadtrees for approximating
Euclidean TSP and TSP with disjoint neighborhoods, respectively. The quadtree variety
that we need, a so-called \emph{stratified grid}, was introduced by Mitchell~\cite{Mi10}
for certain orthogonal polygons. Here we define stratified grids in a more general setting,
for arbitrary geometric graphs (Section~\ref{sec:stratify}).

\section{Preprocessing}\label{sec:pre}

Let $S$ be a set of $n$ disks in the plane. The algorithm first
constructs the graphs $G_1$ and $G_2$ as follows (Fig.~\ref{fig:overview}).

\begin{enumerate} \itemsep 0pt
\item Select an \emph{independent subset} $I$, $I\subseteq S$, of pairwise
  disjoint disks by the following greedy algorithm: Set $I:=\emptyset$.
  Consider the disks in $S$ in increasing order of radius (with ties
  broken arbitrarily), and successively place a disk $D\in S$
  into $I$ if it is disjoint from all previous disks in $I$.
\item Compute a constant-factor approximate TSP tour $\xi_0$ for $I$
using the algorithm in~\cite{BGK+05} or~\cite{EFS09}.
(A PTAS for disjoint disks in the plane is available~\cite{DM03,Sp14}
but not needed here.) It is clear that $\len(\xi_0) = O(\len(\opt))$.
\item Let $R$ be a minimum axis-parallel square such that $\conv(R)$ intersects every
  disk in $S$ (\ie, every disk intersects $\partial R$ or is contained in $R^{\circ}$).
  The square $R$ is determined by up to $3$ disks in $S$, thus $R$ can be trivially computed
  in $O(n^4)$ time: there are $O(n^3)$ squares that pairs and triples define,
  and each can be checked in $O(n)$ time as to whether it intersects all disks.
  Alternatively, finding $R$ is an LP-type problem of combinatorial dimension $3$
  that can be solved in $O(n)$ time~\cite{MSW96}[Section~5];
  see also~\cite{GMRS08} for a modern treatment of LP-type problems and violator spaces.
  Let $r$ denote its side-length of $R$; obviously, we have $r\leq \len(\opt)$.
\item Let $G_1$ be the union of $\xi_0$, $R$, and a shortest line segment connecting
  $\xi_0$ and $R$ (if disjoint).
\end{enumerate}

The graph $G_1$ intersects all disks in $I$ and possibly some disks in
$S\setminus I$.
Our primary interest is in the disks in $S$ that are disjoint from $G_1$.

\begin{figure}[htbp]
\centering
\includegraphics[width=0.99\textwidth]{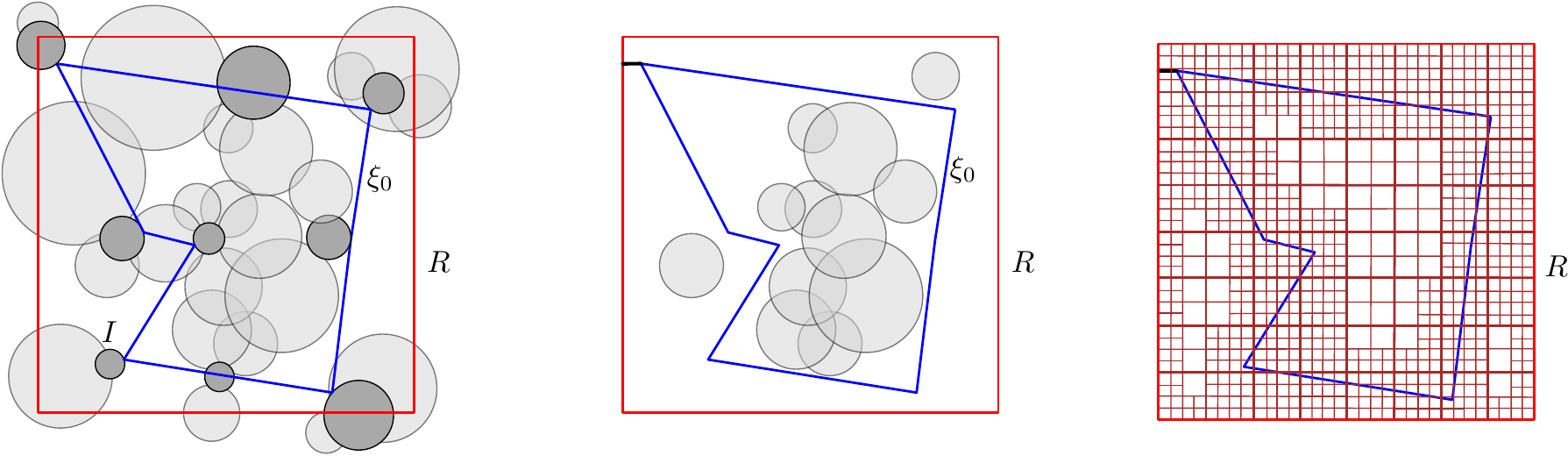}
\caption{Left: a set of disks in the plane; the independent set $I$
  selected by a greedy algorithm is highlighted; a TSP tour $\xi_0$
  for $I$, and a minimum square $R$ intersecting all disks are shown.
Middle: $G_1$ is the union of $\xi_0$, $R$, and a shortest line segment connecting them;
only the disks in $S_2\cup S_3$ that do not intersect $G_1$ are shown.
Right: a stratified grid for $R$ and $G_1$.}
\label{fig:overview}
\end{figure}

\begin{lemma}\label{lem:pre}
For every disk $D\in S$, we have $\dist(D,G_1)\leq \diam(D)$.
\end{lemma}
\begin{proof}
Let $D\in S$. If $G_1$ intersects $D$, then $\dist(D,G_1)=0$, and the
claim is trivial. Assume that $D$ is disjoint from $G_1$. Since
$\xi_0$ intersects every disk in $I$ and $\xi_0\subseteq G_1$, we have
$D\in S\setminus I$. By the greedy choice of $I\subseteq S$, the disk
$D$ intersects some disk $D'\in I$ of equal or smaller radius, where
$G_1$ intersects $D'$. Consequently, $\dist(D,G_1)\leq \diam(D')\leq \diam(D)$.
\end{proof}

\paragraph{Connecting nearby disks to $G_1$.}
We partition $S$ into three subsets:
let $S_1$ be the set of disks in $S$ that intersect $G_1$;
let $S_2$ be the set of disks $D\in S\setminus S_1$ such that
$\dist(D,G_1)\leq \frac{r}{n}$; and let $S_3=S\setminus (S_1\cup S_2)$.
Let $G_2$ be a graph that consists of $|S_2|$ line segments:
specifically for every $D\in S_2$, $G_2$ contains a shortest segment
connecting $D$ and $G_1$. Then $\len(G_2)=\sum_{D\in S_2}
\dist(D,G_1)\leq |S_2|\cdot \frac{r}{n}\leq r\leq \len(\opt)$.
By construction, we have
\begin{equation}\label{eq:s3}
\dist(D,G_1)> \frac{r}{n} \mbox{ \rm for every }D\in S_3.
\end{equation}

By Lemma~\ref{lem:pre} and inequality~\eqref{eq:s3} we have

\begin{corollary}\label{cor:diam}
For every disk $D\in S_3$, we have $\diam(D) > \frac{r}{n}$.
\end{corollary}

In the next section, we show how to find a geometric graph $G_3$
such that $G_3$ intersects every disk in $S_3$ and $G_1\cup G_3$ is
connected (note, however, that $G_3$ need not be connected).

\section{Stratified Grids}\label{sec:stratify}

Recall that we have a geometric graph $G_1$, and a set
$S_3$ of at most $n$ disks in the interior of an axis-aligned square $R$
of side-length $r$, $r \leq \len(\opt)$, satisfying \eqref{eq:s3}.
Let $\opt(S_3,G_1)$ denote a geometric graph $\Gamma$ of
minimum length such that $G_1\cup \Gamma$ is connected and
intersects every disk in $S_3$. Note that $\len(\opt(S_3,G_1))
\leq \len(\opt(S_3)) \leq \len(\opt)$, for every $G_1$.

In Sections~\ref{sec:hitting} and~\ref{sec:chan},
we use hitting sets to compute a $O(1)$-approximation of $\opt(S_3,G_1)$.
Similarly to a quadtree decomposition, we recursively construct a subdivision of $R$
into squares of side-lengths $r/2^i$, for $i=0,1,\ldots ,\lceil \log n\rceil$.
Refer to Fig.~\ref{fig:overview}\,(right).

Previously, Mitchell~\cite{Mi10} used a similar quadtree decomposition for TSPN with
disjoint regions in the plane, coined the term ``stratified grid,'' and derived several
basic properties of quadtrees that we rederive here. Specifically, he proved analogues
of Lemmas~\ref{lem:diam} and~\ref{lem:curve} for the problem studied in~\cite{Mi10}.
However, Mitchell used stratified grids only for special types of orthogonal polygons,
called \emph{histograms}~\cite{L87}; here we generalize this tool to arbitrary geometric graphs.

The following algorithm subdivides a square $Q$ unless it is too small
(\ie, $\diam(Q)<\frac{r}{2n}$)
or it is relatively far from $G_1$ (\ie, $\diam(Q)<\dist(Q,G_1)$).
\begin{quote}{\bf Stratify$(R,G_1)$.}
Let $L$ be a FIFO queue and $\Q$ be a set of axis-aligned squares.
Set $L=(R)$ and $\Q=\emptyset$.
Repeat the following while $L$ is nonempty.
Set $Q\leftarrow \texttt{dequeue}(L)$.
If $\diam(Q)\geq \max(\frac{r}{2n},\dist(Q,G_1))$, then subdivide $Q$
into four congruent axis-aligned squares, and enqueue them onto $L$.
Otherwise, let $\Q \leftarrow \Q \cup \{Q\}$. Return $\Q$.
\end{quote}

It is worth noting that $\Q$ does not directly depend on the disks in $S_3$,
but only indirectly, via $G_1$.
By construction, the squares in $\Q$ are interior-disjoint, and every
square in $\Q$ has diameter at least $r/(4n)$.
Consequently, the number of squares in $\Q$ is $O(n^2)$.
Thus the algorithm {\bf Stratify}$(R,G_1)$ runs in polynomial
time in $n$, since $O(n^2)$ squares are enqueued onto $L$,
and $\dist(Q,G_1)$ can be computed in polynomial time for all $Q\in L$.
We show that the squares in $\Q$ have a property similar to the disks in $S_3$
(cf. Lemma~\ref{lem:pre}): only larger squares can be farther from $G_1$.

\begin{lemma}\label{lem:diam}
For every square $Q\in \Q$, we have $\dist(Q,G_1)\leq 3\ \diam(Q)$.
\end{lemma}
\begin{proof}
Put $q=\diam(Q)$. Recall that $Q$ is obtained by subdividing a square $Q'$,
$Q\subset Q'$, with $\diam(Q')=2q$. Since $Q'$ is subdivided by the
algorithm, we have $\diam(Q') \geq \max(\frac{r}{2n},\dist(Q',G_1))$.
Since $\dist(p',Q) \leq q$ for every point $p' \in Q'$,
the triangle inequality yields $\dist(Q,G_1) \leq 3q$.
\end{proof}

\begin{figure}[htbp]
\centering
\includegraphics[width=0.85\textwidth]{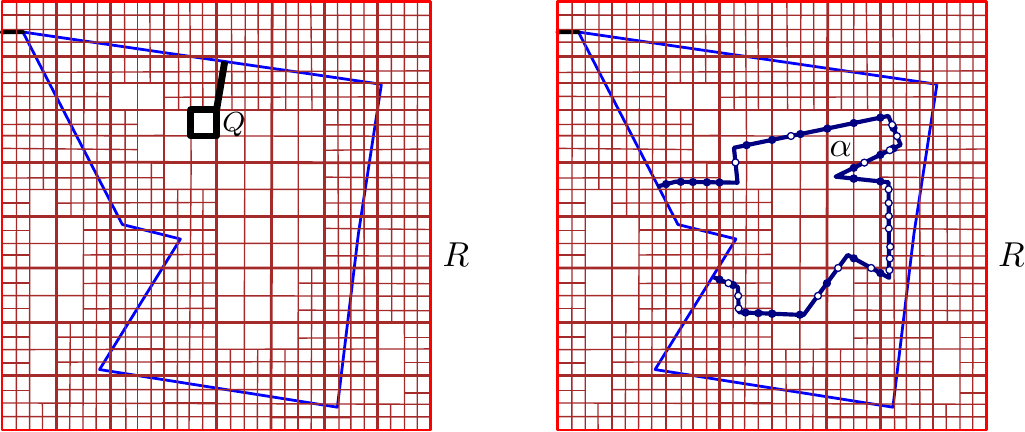}
\caption{Left: a square $Q$ of the stratified grid, and the graph $\gamma(Q)$.
Right: a polygonal curve $\alpha$; the intersections of $\alpha$ with
horizontal (resp., vertical) edges of the stratified grid are marked with empty
(resp., full) dots.}
\label{fig:gamma}
\end{figure}

For every square $Q\in \Q$, we define a graph $\gamma(Q)$ that consists
of the boundary of $Q$ and a shortest line segment from $Q$ to $G_1$;
see Fig.~\ref{fig:gamma}\,(left). By Lemma~\ref{lem:diam}, we have
$\len(\gamma(Q))\leq (3+2\sqrt{2}) \, \diam(Q)$; on the other hand,
$\len(\gamma(Q)) \geq \per(Q) =2\sqrt{2} \, \diam(Q)$,
and so we have the following.
\begin{corollary}\label{cor:gamma}
For every $Q\in \Q$, we have $\len(\gamma(Q))=\Theta(\diam(Q))$.
\end{corollary}

The following observation is crucial for reducing the problem of
approximating $\opt(S_3,G_1)$ to a minimum weight hitting set problem.

\begin{lemma}\label{lem:gamma}
  If a square $Q\in \Q$ intersects a disk $D\in S_3$, then
\begin{enumerate}\itemsep -2pt
\item[{\rm (i)}]  $\diam(Q) \leq 2\, \diam(D)$, and
\item[{\rm (ii)}] $D$ intersects the boundary of $Q$ (and the graph $\gamma(Q)$ in particular).
\end{enumerate}
\end{lemma}
\begin{proof}
(i)~Since $Q\in \Q$, Algorithm {\bf Stratify}$(R,G_1) $ did not subdivide $Q$, and so
we have $\diam(Q)< \frac{r}{2n}$ or $\diam(Q)<\dist(Q,G_1)$.
If $\diam(Q)< \frac{r}{2n}$, then Corollary~\ref{cor:diam} yields
$$ \diam(Q)< \frac{r}{2n}< \frac{r}{n} < \diam(D) < 2\, \diam(D). $$

If $\diam(Q)<\dist(Q,G_1)$, then $\dist(Q,G_1) \leq \dist(D,G_1) + \diam(D)$
follows from the intersection condition and the triangle inequality.
Consequently,
$$ \diam(Q) < \dist(Q,G_1) \leq \dist(D,G_1) + \diam(D) \leq 2\, \diam(D), $$
where the last inequality holds by Lemma~\ref{lem:pre}.

\smallskip
(ii)~Suppose, to the contrary, that the boundary of $Q$ is disjoint from
$D$, hence $D$ lies in the interior of $Q$. This immediately implies
\begin{equation}\label{eq:interior}
\dist(Q,G_1)\leq \dist(D,G_1).
\end{equation}

Since $Q\in \Q$, Algorithm {\bf Stratify}$(R,G_1) $ did not subdivide $Q$, and so
we have $\diam(Q)< \frac{r}{2n}$ or $\diam(Q)<\dist(Q,G_1)$.
If $\diam(Q)< \frac{r}{2n}$, Corollary~\ref{cor:diam} yields
$\diam(Q)<\frac{r}{2n}< \frac{r}{n}<\diam(D)$.
If $\diam(Q)<\dist(Q,G_1)$, then the combination of \eqref{eq:interior} and
Lemma~\ref{lem:pre} yield
$$ \diam(Q) < \dist(Q,G_1) \leq \dist(D,G_1) < \diam(D). $$

In both cases, we have shown that $\diam(Q)< \diam(D)$. Therefore $D$ cannot lie in
the interior of $Q$, which contradicts the assumption.
\end{proof}

Recall that the squares in $\Q$ can only intersect at common boundary points;
we call such squares \emph{adjacent}.

\begin{lemma}\label{lem:adjacent}
If two squares $Q_1,Q_2\in \Q$ are adjacent and $\diam(Q_1) \leq \diam(Q_2)$, then
$$ \frac{1}{2}\diam(Q_2) \leq \diam(Q_1)\leq \diam(Q_2) \leq 2\ \diam(Q_1). $$
\end{lemma}
\begin{proof}
  If $\diam(Q_1)=\diam(Q_2)$, the inequalities are satisfied.
  We may thus assume that $\diam(Q_1)<\diam(Q_2)=q$.
  Then Algorithm {\bf Stratify}$(R,G_1)$
subdivided a square $Q_1'$ such that $Q_1\subset Q_1'$,
$\diam(Q_1') \leq q$, and $Q_1'\cap Q_2\neq \emptyset$.
The algorithm subdivided $Q_1'$ but did not subdivide $Q_2$. This implies
$$ q \geq \max\left(\frac{r}{2n},\dist(Q_1',G_1)\right) \text{ and }
q< \max\left(\frac{r}{2n},\dist(Q_2,G_1)\right). $$
The first inequality yields
$q \geq \max(\frac{r}{2n},\dist(Q_1',G_1)) \geq \frac{r}{2n}$,
and then the second inequality yields
$q< \max(\frac{r}{2n},\dist(Q_2,G_1)) = \dist(Q_2,G_1)$.
Consequently, $\dist(Q_1',G_1) \leq q< \dist(Q_2,G_1)$.

Since $Q_1'$ and $Q_2$ intersect, and their diameters are at most $q$,
the triangle inequality yields $|\dist(Q_1',G_1)-\dist(Q_2,G_1)|\leq q$.
It follows that
$$ q< \dist(Q_2,G_1) \leq \dist(Q_1',G_1) + q \leq 2q. $$
Similarly, since $Q_1$ and $Q_2$ intersect,
$\dist(Q_1,G_1) \geq \dist(Q_2,G_1) -\diam(Q_1) > q-\diam(Q_1)$.
Combining with Lemma~\ref{lem:diam}, we get
$$ 3 \, \diam(Q_1) \geq \dist(Q_1,G_1) > q-\diam(Q_1), $$
that is, $\diam(Q_1)>q/4$. Finally, recall that the ratio between the diameters
of any two squares in $\Q$ is a power of 2. Therefore
$q/4<\diam(Q_1)<q$ yields $\diam(Q_1)=q/2$, as required.
\end{proof}

\section{Hitting Sets for Squares and Disks}\label{sec:hitting}

For the graph $G_1$ and the set of disks $S_3$, we define a hypergraph
$\G=(\Q,E)$, where the vertex set is the set $\Q$ of squares in the
stratified grid; and for every disk $D\in S_3$, the set of squares in $\Q$
that intersect $D$ forms a hyperedge in $E$. Thus, a subset $\H\subseteq \Q$
of squares is a \emph{hitting set} in the hypergraph $\G$ if and only if
every disk in $S_3$ intersects some square in $\H$.

For every hitting set $\H$, the geometric graph $\Gamma=\cup_{Q\in \H} \gamma(Q)$
intersects every disk in $S_3$ by Lemma~\ref{lem:gamma}, and $G_1\cup \Gamma$
is connected by construction. Let the \emph{weight} of a square $Q\in \Q$ be
$w(Q)=\diam(Q)$. In this section, we show that the minimum-weight hitting
set for $\G=(\Q,E)$ is a $O(1)$-approximation for $\opt(S_3,G_1)$. The
following technical lemma considers a single curve (i.e., a Jordan arc).
For a curve $\alpha$, let $\Q(\alpha)$ denote the set of squares in $\Q$ that
intersect $\alpha$. Refer to Fig.~\ref{fig:gamma}\,(right).

\begin{lemma}\label{lem:curve}
Let $\alpha$ be a directed polygonal curve whose start and end points lie on $G_1$.
If $\alpha$ intersects at least one disk in $S_3$, then
$\len(\alpha) = \Omega(\sum_{Q\in \Q(\alpha)} \diam(Q))$.
\end{lemma}
\begin{proof}
By~\eqref{eq:s3}, we have $\dist(D,G_1)> \frac{r}{n}$ for every disk $D\in S_3$.
Since $\alpha$ intersects at least one disk in $S_3$, we have
$\len(\alpha)\geq \frac{2r}{n}$.

Let $A=(Q_0,Q_1,\ldots , Q_m)$ be the sequence of distinct squares
that intersect $\alpha$ in the order in which they are first
encountered by $\alpha$ (with no repetitions and ties broken arbitrarily).
Since $Q_0$ intersects $G_1$, we have $\diam(Q_0) < \max(\frac{r}{2n},0)=\frac{r}{2n}$,
and consequently
\begin{equation}\label{eq:alpha0}
\len(\alpha)\geq \frac{2r}{n} \geq 4 \, \diam(Q_0).
\end{equation}

Let $B=(Q_{\sigma(0)}, Q_{\sigma(1)},\ldots , Q_{\sigma(\ell)})$ be the
subsequence of $A$ such that $\sigma(0)=0$ and a square $Q_i$, $1\leq i\leq m$,
is added to $B$ if it is disjoint from $Q_j$ for all $0\leq j<i$.
By construction, $B$ consists of pairwise disjoint squares,
and every square in $A$ is either in $B$ or adjacent to some square in $B$.
By Lemma~\ref{lem:adjacent}, the sizes of adjacent squares in $\Q$ differ by a factor
of at most 2. Consequently, each square in $\Q$ is adjacent to
at most 12 squares in $\Q$ (at most two along each side and at most one at each corner).
It follows that
\begin{equation}\label{eq:alpha1}
\sum_{i=0}^m\diam(Q_i) = \Theta \left(\sum_{j=0}^{\ell} \diam(Q_{\sigma(j)})\right).
\end{equation}

For $j=0,\ldots , \ell$, let $p_{\sigma(j)}$ be the first intersection point of
$\alpha$ with $Q_{\sigma(j)}$. For two points $p,q\in \alpha$,
denote by $\alpha(p,q)$ the portion of $\alpha$ between $p$ and $q$.
Since the squares in $B$ are pairwise disjoint, and the sizes of adjacent squares
differ by at most a factor of 2 (Lemma~\ref{lem:adjacent}), we have
\begin{equation}\label{eq:alpha2}
\len\left(\alpha(p_{\sigma(j)},p_{\sigma(j+1)})\right)
\geq |p_{\sigma(j)}p_{\sigma(j+1)}|
\geq \frac{1}{2\sqrt{2}} \max\left(\diam(Q_{\sigma(j)}),\diam(Q_{\sigma(j+1)}) \right)\nonumber
\end{equation}
for $j=0,\ldots, \ell-1$. Consequently, if $\ell\geq 1$, we have
\begin{equation}\label{eq:alpha3}
\len(\alpha)
= \sum_{j=0}^{\ell-1} \len\left(\alpha(p_{\sigma(j)},p_{\sigma(j+1)})\right)
= \Omega\left(\sum_{j=0}^{\ell}\diam(Q_{\sigma(j)})\right).
\end{equation}
The combination of \eqref{eq:alpha0}, \eqref{eq:alpha1}, and \eqref{eq:alpha3}
yields $\len(\alpha) =\Omega(\sum_{i=0}^m \diam(Q_i))$, as required.
\end{proof}

\begin{lemma}\label{lem:approx}
If $\Gamma$ is a geometric graph such that $\Gamma$ intersects
every disk in $S_3$ and $G_1\cup \Gamma$ is connected, then there is a
hitting set $\H\subseteq \Q$ for $\G=(\Q,E)$ such that
\begin{equation}\label{eq:hit}
\len(\Gamma) =\Omega\left(\sum_{Q\in \H} \diam(Q)\right).
\end{equation}
\end{lemma}
\begin{proof}
Let $\H$ be the set of squares in $\Q$ that intersect $\Gamma$, and observe
that $\H$ is a hitting set for $\G$.
For each connected component $C$ of $\Gamma$, let $\alpha$ be a directed polygonal
curve that starts and ends at some points in $G_1\cap C$ and traverses every
edge of $C$ at least once and at most twice.
Then $\len(C) \geq \frac{1}{2}\len(\alpha)$,
and $\len(\alpha)=\Omega(\sum_{Q\in \Q(\alpha)}\diam(Q))$ by Lemma~\ref{lem:curve}.
Summation over all the components of $\Gamma$ yields~\eqref{eq:hit}.
\end{proof}

Recall that $\opt(S_3,G_1)$ is a geometric graph $\Gamma$ of minimum length
such that $G_1\cup \Gamma$ is connected and intersects every disk in $S_3$.
Let $W_0$ denote the minimum weight of a hitting set in the hypergraph $\G$.
The main result of this section is the following.

\begin{corollary}\label{cor:approx}
We have $W_0= O(\len(\opt(S_3,G_1)))$.
\end{corollary}
\begin{proof}
Invoke Lemma~\ref{lem:approx} with $\Gamma=\opt(S_3,G_1)$.
Then $\G$ has a hitting set $\H\subset \Q$ of weight
$\sum_{Q\in  \H}\diam(Q) = O(\len(\Gamma))=O(\len(\opt(S_3,G_1)))$.
This is clearly an upper bound on the minimum weight $W_0$ of a hitting set in $\G$.
\end{proof}

\section{Hitting Sets for Points and Disks}\label{sec:chan}

In Section~\ref{sec:hitting} we defined a hypergraph $\G=(\Q,E)$ for
squares and disks; that is, the vertices are squares in $\Q$ and the
hyperedges are the squares intersecting a disk in $S_3$. In order to
apply Theorem~\ref{thm:chan} by Chan~\etal~\cite{CGKS12}, we reduce
the problem to a traditional geometric hypergraph problem, where the
vertices are points in $\RR^2$ and a hyperedge corresponds to the set of points
contained in a disk $D\in S_3$.

For each square $Q\in \Q$, we define a set of 25 \emph{sentinel}
points, and show (Lemma~\ref{lem:sentinel}) that if a disk $D\in S_3$
intersects $Q$, then $D$ contains one of the sentinel points of $Q$.
A constant number of sentinels suffice if none of the disks intersecting $Q$
is too small, and indeed, Lemma~\ref{lem:gamma} has shown that this is the case.

For a square $Q\in \Q$, where $Q=[a,a+h]\times [b,b+h]$, let the 25 sentinel points be
$(a+ih/2,b+jh/2)$ for all $i,j\in\{-1,0,1,2,3\}$; see Fig.~\ref{fig:sentinel}\,(left).

\begin{figure}[htbp]
\centering
\includegraphics[width=0.75\textwidth]{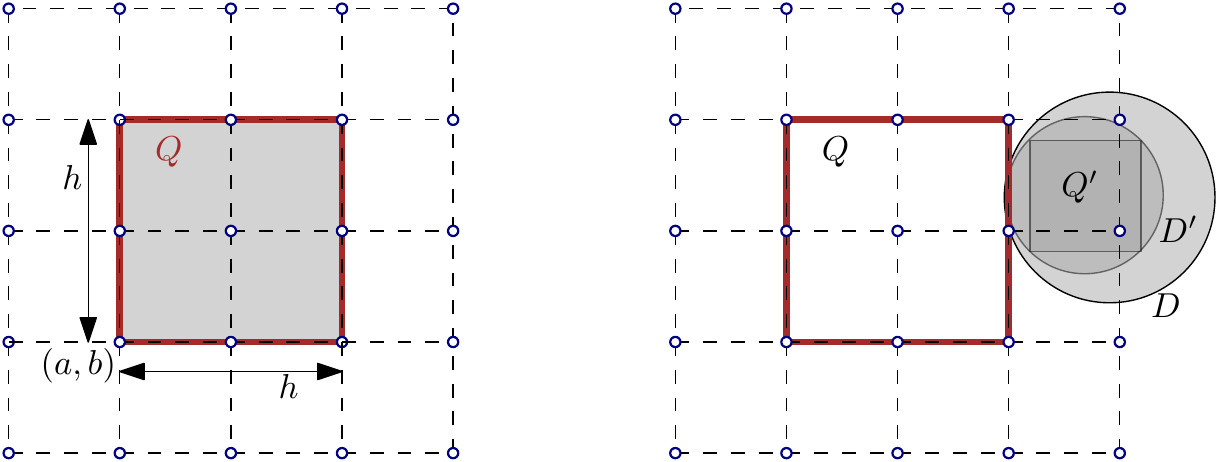}
\caption{Left: the set of 25 sentinel points for a square $Q$.
Right: a disk $D\in S_3$ intersects a square $Q\in \Q$.}
\label{fig:sentinel}
\end{figure}

\begin{lemma}\label{lem:sentinel}
If a disk $D\in S_3$ intersects a square $Q\in \Q$, then $D$ contains
a sentinel point corresponding to $Q$.
\end{lemma}
\begin{proof}
Assume that a disk $D\in S_3$ intersects a square $Q\in \Q$ of side
length $h$; refer to Fig.~\ref{fig:sentinel}\,(right).
By Lemma~\ref{lem:gamma}, $\diam(Q) \leq 2\, \diam(D)$. By scaling down $D$
from an arbitrary center in $D\cap Q$, we find a disk $D'$ intersecting $Q$ with
$\diam(D')=\frac{1}{2}\, \diam(Q)$. The inscribed axis-aligned square
$Q'$ of $D'$ has $\diam(Q')=\frac{1}{2}\, \diam(Q)$.
That is, the side-length of $Q'$ is $h/2$, and $\dist(Q',Q)\leq (\sqrt{2}-1)h$.
Since the sentinels of $Q$ form a section of a square lattice of (the
same) side-length $h/2$, within distance $\frac{\sqrt2}{2}\, h$ from $Q$,
some sentinel of $Q$ lies in $Q'$, and hence in $D'\subseteq D$, as claimed.
\end{proof}

We define a new weighted hypergraph $\G'=(V',E')$, where $V'$ is the
union of sentinel point sets for all $Q\in \Q$ that lie in $R$
(sentinels in the exterior of $R$ are discarded);
and each hyperedge in $E'$ is the set of sentinels in $V'$ contained in a disk $D\in S_3$.
Note that a sentinel $s\in V'$ may correspond to several squares in $\Q$.
Let the \emph{weight} of a sentinel $s\in V'$ be the sum of the diameters of the
squares $Q\in \Q$ that correspond to $s$. Hence the total weight of all sentinels
is at most $25\sum_{Q\in \Q} \diam(Q)$. We next derive a bound on the weight of
each sentinel.

\begin{lemma}\label{lem:weight}
For every $Q\in \Q$, the weight of every sentinel corresponding to $Q$ is $O(\diam(Q))$.
\end{lemma}
\begin{proof}
By Lemma~\ref{lem:adjacent}, the side-lengths of adjacent squares of the stratified grid differ
by at most a factor of $2$. Consequently, every sentinel in $V'$ corresponding to a square $Q\in \Q$
is contained in $Q$ or in a square of $\Q$ adjacent to $Q$.

Let $s\in V'$ be a sentinel.
Then $s$ may correspond to all squares in $\Q$ that contain $s$, and to adjacent squares in $\Q$.
Every point is contained in at most $4$ squares of $\Q$, whose side-lengths differ by a factor
of at most $2$; and they are each adjacent to $O(1)$ additional squares whose side-lengths
differ by another factor of at most $2$. Overall, $s$ corresponds to $O(1)$ squares in $\Q$
whose side-lengths differ by a factor of $\Theta(1)$.
Therefore, the weight of $s$ is $O(\diam(Q))$ for every square $Q\in Q$ corresponding to $s$.
\end{proof}

By Theorem~\ref{thm:chan}, there is a polynomial-time
$O(1)$-approximation algorithm for MWHS on $\G'=(V',E')$. It remains to
show that a $O(1)$-approximation for MWHS on the hypergraph $\G'=(V',E')$ provides
a $O(1)$-approximation for MWHS on the hypergraph $\G=(\Q,E)$.

\begin{lemma}\label{lem:prime}
\begin{enumerate}\itemsep 0pt
\item For every hitting set $\H\subseteq \Q$ for $\G$, the set
  $H'$ of sentinels in $V'$ corresponding to the squares $Q\in \H$ is
  a hitting set for $\G'$ of weight $O(\sum_{Q\in \H} \diam(Q))$.
\item For every hitting set $H'\subseteq V'$ for $\G'$, the set $\H$ of
  squares $Q\in \Q$ that contain the sentinel points in $H'$ is a hitting set
  for $\G$ of weight $O(\sum_{s\in H'} w(s))$.
\end{enumerate}
\end{lemma}
\begin{proof}
(1) If $\H$ is a hitting set for $\G$, then every disk $D\in S_3$
    intersects some square $Q\in \H$. By Lemma~\ref{lem:sentinel},
    $D$ contains one of the sentinels of $Q$.
    Consequently, every disk $D\in S_3$ contains a sentinel in $H'$.
    Every square $Q\in \H$ corresponds to 25 sentinels, each of weight $O(\diam(Q)$
    by Lemma~\ref{lem:weight}. The weight of $H'$ is $O(\sum_{Q\in \H} \diam(Q))$.

(2) If $H'$ is a hitting set for $\G'$, then every disk $D\in S_3$
    contains some point $s\in H'$.
    The point $s$ lies in a square $Q\in \Q$ of the stratified grid, which
    is in $\H$. Consequently, every disk $D\in S_3$ intersects some square
    $Q\in \H$. By construction, the weight of each sentinel $s$ is the sum of
    weights of the corresponding squares in $\Q$, including all squares
    in $\Q$ that contain $s$. Therefore, the weight of $\H$ is at most
    $\sum_{s\in H'}w(s)$, as required.
\end{proof}

We are now ready to prove Theorem~\ref{thm:disk} by
analyzing the constructed graph $G=G_1\cup G_2\cup G_3$.

\begin{proof}[Proof of Theorem~\ref{thm:disk}.]
Let $S$ be a set of $n$ disks in $\RR^2$. Compute an independent set
$I\subset S$ as described in Section~\ref{sec:pre},
and a TSP tour $\xi_0$ for $I$ with $\len(\xi_0)=O(\len(\opt))$
(as in~\cite{BGK+05} or~\cite{EFS09}). Compute the graph $G_1$ with
$\len(G_1)=O(\len(\opt))$, and the partition $S=S_1\cup S_2\cup S_3$
as described in Section~\ref{sec:pre}.
The graph $G_1$ intersects the disks in $S_1$.

Construct the graph $G_2$, which contains a shortest segment between
$G_1$ and every disk $D\in S_2$.
The length of this graph is $\len(G_2)=\sum_{D\in S_2}
\dist(D,G_1)\leq |S_2|\cdot \frac{r}{n}\leq r\leq \len(\opt)$.

Compute the stratified grid $\Q$, and construct the weighted hypergraph $\G'=(V',E')$,
where $V'$ is the set of sentinel points for all squares $Q\in \Q$, the weight of
a sentinel $s$ is the sum of diameters of the corresponding squares $Q\in \Q$,
and for every disk $D\in S_3$, the set of sentinels lying in $D$ forms
a hyperedge in $E'$. Use the algorithm by Chan~\etal~\cite{CGKS12}
to compute a hitting set $H'$ for $\G'$ whose weight is $O(1)$ times the minimum.
Let $\H$ be the set of squares in $\Q$ containing the sentinels
in $H'$. By Lemma~\ref{lem:prime}, $\H$ is a hitting set for the
hypergraph $\G=(\Q,E)$ whose weight is at most $O(1)$ times the minimum
$W_0$. Put $G_3= \cup_{Q\in \H}\gamma(Q)$.
Then $G_3$ intersects every disk in $S_3$ by Lemma~\ref{lem:gamma}, and
$\len(G_3)=\Theta(\sum_{Q\in \H}\diam(Q))=\Theta(W_0)$ by Corollary~\ref{cor:gamma}.
Finally, Corollary~\ref{cor:approx} yields
$\len(G_3)=O(\len(\opt(S_3,G_1)))$.
By the definition of $\opt(S_3,G_1))$, we have
$\len(\opt(S_3,G_1)) \leq \len(\opt(S_3)) \leq \len(\opt)$,
and consequently $\len(G_3)=O(\len(\opt))$.

Note that the graph $G=G_1\cup G_2\cup G_3$ is connected by
construction, it intersects every disk in $S=S_1\cup S_2\cup S_3$,
and $\len(G)=\len(G_1)+\len(G_2)+\len(G_3)=O(\len(\opt))$.
Consequently, an Eulerian tour of the multi-graph containing each edge of
$G$ twice visits each disk and its length is $2 \, \len(G) = O(\len(\opt))$,
as required.

Since the above steps as well as algorithm {\bf Stratify}$(R,G_1)$
all run in time that is polynomial in $n$, the constant-factor approximation algorithm
for TSP with disks runs in polynomial time.
\end{proof}

\section{Conclusions}  \label{sec:conclusion}

In this paper, we obtained the first constant-ratio approximation for TSP
with disks in the plane. This is the first result of this kind for a class
of planar convex bodies of arbitrary size that can intersect in an arbitrary fashion.
In light of the connection we established between TSPN and MWHS in geometric hypergraphs,
the following question emerges:

\begin{enumerate} \itemsep 0pt

\item Besides regions of linear union complexity (e.g., disks and
  pseudo-disks\footnote{A set of regions $\R$ consists of \emph{pseudo-disks},
    if every pair of regions $\omega_1,\omega_2 \in \R$
     satisfies the \emph{pseudo-disk property}: the sets $\omega_1
     \setminus \omega_2$ and $\omega_2 \setminus \omega_1$
are connected~\cite[p.~293]{BCK+08}. Equivalently, the boundaries $\partial \omega_1$
and $\partial \omega_2$ have at most two proper intersection points.}),
what other types of regions admit a constant-factor approximation for
the minimum weight hitting set problem?

\end{enumerate}

Obviously, a constant-factor approximation for MWHS with a certain type of
neighborhoods does not automatically imply a constant-factor approximation
for TSPN with the same type of neighborhoods.
We conclude with a few, perhaps the simplest still unsolved
questions on TSPN that we could identify:

\begin{enumerate} \itemsep 0pt
  \setcounter{enumi}{1}

 \item Is there a constant-factor approximation algorithm for TSP with a set
   of objects of linear union complexity, \eg,  pseudo-disks?

\item Is there a constant-factor approximation algorithm for TSP with
  convex bodies in the plane?\footnote{Very recently, Mitchell~\cite{Mi16b}
  proposed a constant-factor approximation algorithm for this variant.
  However, no complete proof is available at the time of this writing.
  In fact, we believe that TSP with planar convex bodies is much harder to approximate
  than TSP with disks.}

\item Is TSP with disks in the plane APX-hard? Is TSP with convex bodies in the plane
  APX-hard?

\item Is there a constant-factor approximation algorithm for TSP with balls
 (with arbitrary radii and intersections) in $\RR^d$, in fixed dimension $d\geq 3$?

\end{enumerate}


\begin{thebibliography}{99}

\bibitem{AES12}
P.~K. Agarwal, E.~Ezra, and M.~Sharir,
Near-linear approximation algorithms for geometric hitting sets,
\emph{Algorithmica} {\bf 63(1-2)} (2012), 1--25.

\bibitem{AH94}
E.~M. Arkin and R.~Hassin,
Approximation algorithms for the geometric covering salesman problem,
\emph{Discrete Applied Mathematics} \textbf{55(3)} (1994), 197--218.

\bibitem{Ar98} S.~Arora,
Polynomial time approximation schemes for Euclidean traveling salesman
and other geometric problems,
\emph{Journal of the ACM} \textbf{45(5)} (1998), 753--782.

\bibitem{BGK12}
Y.~Bartal, L.-A.~Gottlieb, and R.~Krauthgamer,
The traveling salesman problem: low-dimensionality implies a polynomial time
approximation scheme,
in \emph{Proc. 44th Symposium on Theory of Computing (STOC)},
ACM Press, 2012, pp.~663--672.

\bibitem{BCK+08} M.~de~Berg, O. Cheong, M. van Kreveld, and M.~Overmars,
\emph{Computational Geometry}, 3rd edition, Springer, Heidelberg, 2008.

\bibitem {BGK+05} M. de Berg, J. Gudmundsson, M. J. Katz,
C. Levcopoulos, M. H. Overmars, and A. F. van der Stappen,
TSP with neighborhoods of varying size,
\emph{Journal of Algorithms} \textbf{57(1)} (2005), 22--36.

\bibitem{BE97} M. Bern and D. Eppstein,
Approximation algorithms for geometric problems,
in \emph{Approximation Algorithms for NP-hard Problems (D. S. Hochbaum, ed.)},
PWS Publishing Company, Boston, MA, 1997, pp.~296--345.

\bibitem{BFG+09}
H. L. Bodlaender, C. Feremans, A. Grigoriev, E. Penninkx, R. Sitters, and T. Wolle,
On the minimum corridor connection problem and other generalized geometric problems,
\emph{Computational Geometry: Theory and Applications} \textbf{42(9)} (2009), 939--951.

\bibitem{BG95}
H. Br\"onnimann and M. T. Goodrich,
Almost optimal set covers in finite VC-dimension,
\emph{Discrete \& Computational Geometry} \textbf{14(4)} (1995), 463--479.

\bibitem{BGM+15}
N.~Bus, S.~Garg, N.~H. Mustafa, and S.~Ray,
Improved local search for geometric hitting set,
in \emph{Proc. 32nd Symposium on Theoretical Aspects of Computer Science (STACS)},
vol.~30 of LIPIcs, Schloss Dagstuhl, 2015, pp.~184--196.

\bibitem{CE11}
T.-H. H. Chan and K. Elbassioni,
A QPTAS for TSP with fat weakly disjoint neighborhoods in doubling metrics,
\emph{Discrete \& Computational Geometry} \textbf{46(4)} (2011), 704--723.

\bibitem{CJ14}
T.-H. H. Chan and S. H.-C. Jiang,
A PTAS for TSP with fat weakly disjoint neighborhoods in doubling metrics,
technical report, e-prints for the Optimization Community,
\url{http://www.optimization-online.org/DB_HTML/2014/11/4649.html}.

\bibitem{CGKS12}
T.M.~Chan, E. Grant, J. K\"onemann, and M. Sharpe,
Weighted capacitated, priority, and geometric set cover via improved
quasi-uniform sampling,
in \emph{Proc. 23rd ACM-SIAM Symposium on Discrete Algorithms (SODA)},
SIAM, 2012, pp. 1576--1585.

\bibitem{CV07}
K. L. Clarkson and K. R. Varadarajan,
Improved approximation algorithms for geometric set cover,
\emph{Discrete \& Computational Geometry} {\bf 37(1)} (2007), 43--58.

\bibitem {DO08} M. Dror and J. B. Orlin,
Combinatorial optimization with explicit delineation of the ground set
by a collection of subsets,
\emph{SIAM Journal on Discrete Mathematics}
\textbf{21(4)} (2008), 1019--1034.

\bibitem {DM03}
A. Dumitrescu and J.~S.~B. Mitchell,
Approximation algorithms for TSP with neighborhoods in the plane,
\emph{Journal of Algorithms} \textbf{48(1)} (2003), 135--159.

\bibitem {DT15}
A. Dumitrescu and Cs. D. T\'oth,
On the total perimeter of homothetic convex bodies in a convex container,
\emph{Beitr\"age zur Algebra und Geometrie} \textbf{56(2)} (2015), 515--532.

\bibitem{DT16}
A. Dumitrescu and Cs. D. T\'oth,
The traveling salesman problem for lines, balls and planes,
\emph{ACM Transactions on Algorithms} \textbf{12(3)} (2016), article~43.

\bibitem {EFS09}
K. M. Elbassioni, A. V. Fishkin, and R. Sitters,
Approximation algorithms for the Euclidean traveling salesman problem
with discrete and continuous neighborhoods,
\emph{International Journal on Computational Geometry \& Applications}
\textbf{19(2)} (2009), 173--193.

\bibitem{ERS05}
G. Even, D. Rawitz, and S. Shahar,
Hitting sets when the VC-dimension is small,
\emph{Information Processing Letters} \textbf{95(2)} (2005), 358--362.

\bibitem{GGJ76}
M. R. Garey, R. Graham, and D. S. Johnson,
Some NP-complete geometric problems,
in \emph{Proc. 8th ACM Symposium on Theory of Computing (STOC)},
ACM Press, 1976, pp.~10--22.

\bibitem{GJ79}
M. R. Garey and D. S. Johnson:
\emph{Computers and Intractability: A Guide to the Theory of NP-Completeness},
W. H. Freeman and Company, New York, 1979.

\bibitem{GMRS08}
B. G\"artner, J. Matou\v{s}ek, L. R\"ust, and P. \v{S}kovro\v{n},
Violator spaces: Structure and algorithms,
\emph{Discrete Appl. Math.} \textbf{156(11)} (2008), 2124--2141.

\bibitem{GL99}
J.~Gudmundsson and C.~Levcopoulos,
A fast approximation algorithm for {TSP} with neighborhoods,
\emph{Nordic Journal of Computing} \textbf{6(4)} (1999), 469--488.

\bibitem{KS13}
P. Kamousi and S. Suri,
Euclidean traveling salesman tours through stochastic neighborhoods,
in \emph{Proc. 24th International Symposium on Algorithms and Computation (ISAAC)},
LNCS~8283, Springer, 2013, pp.~644--654.

\bibitem{KLPS86}
K. Kedem, R. Livne, J. Pach, and M. Sharir,
On the union of Jordan regions and collision-free translational motion
amidst polygonal obstacles,
\emph{Discrete \& Computational Geometry} \textbf{1(1)} (1986), 59--70.

\bibitem{L87}
C. Levcopoulos, Heuristics for Minimum Decompositions of Polygons,
PhD thesis, Link\"oping Studies in Science and Technology, No. 74, 1987.

\bibitem{LL84}
C.~Levcopoulos and A.~Lingas,
Bounds on the length of convex partitions of polygons,
in \emph{Proc. 4th Conf. on Foundations of Software Technology and
  Theoretical Computer Science}, LNCS~181, Springer, 1984, pp.~279--295.

\bibitem{MM95}
C.~Mata and J.~S.~B. Mitchell,
Approximation algorithms for geometric tour and network design problems,
in \emph{Proc. 11th ACM Symposium on Computational Geometry (SOCG)},
ACM Press, 1995, pp.~360--369.

\bibitem{MSW96}
J. Matou\v{s}ek, M. Sharir, and E. Welzl,
A subexponential bound for linear programming,
\emph{Algorithmica} \textbf{16(4)} (1996), 498--516.

\bibitem{Mi99}
J.~S.~B. Mitchell,
Guillotine subdivisions approximate polygonal subdivisions:
A simple polynomial-time approximation scheme for geometric {TSP},
{$k$-MST}, and related problems,
\emph{SIAM Journal on Computing} \textbf{28(4)} (1999), 1298--1309.

\bibitem{Mi00}
J.~S.~B. Mitchell,
Geometric shortest paths and network optimization,
in \emph{Handbook of Computational Geometry (J.-R. Sack, J.~Urrutia, eds.)},
Elsevier, 2000, pp.~633--701.

\bibitem{Mi04}
J.~S.~B. Mitchell,
Shortest paths and networks,
in \emph{Handbook of Computational Geometry (J. E. Goodman and J. O'Rourke, eds.)},
Chapman \& Hall/CRC, 2004, pp.~607--641.

\bibitem{Mi07}
J.~S.~B. Mitchell,
A PTAS for TSP with neighborhoods among fat regions in the plane,
in \emph{Proc. 18th ACM-SIAM Symposium on Discrete Algorithms (SODA)},
SIAM, 2007, pp.~11--18.

\bibitem{Mi10}
J.~S.~B. Mitchell,
A constant-factor approximation algorithm for TSP with pairwise-disjoint
connected neighborhoods in the plane,
in \emph{Proc. 26th Symposium on Computational Geometry (SOCG)},
ACM Press, 2010, pp.~183--191.

\bibitem{Mi16a}
J.~S.~B. Mitchell,
Clarifying remarks added to ``Approximation algorithms for TSP with neighborhoods in the plane,''
manuscript, \url{http://www.ams.sunysb.edu/~jsbm/papers/tspn-jalgorithms-updated.pdf},
accessed in Feb.~2016.

\bibitem{Mi16b}
J. S. B. Mitchell,
Updated version of ``A constant-factor approximation algorithm for TSP with
pairwise-disjoint connected neighborhoods in the plane,''
manuscript, \url{http://www.ams.sunysb.edu/~jsbm/papers/tspn-socg10-updated.pdf},
accessed in Feb.~2016.

\bibitem{MR10}
N.~H. Mustafa and S.~Ray,
Improved results on geometric hitting set problems,
\emph{Discrete \& Computational Geometry} {\bf 44(4)} (2010), 883--895.

\bibitem {PA95}
J. Pach and P. K. Agarwal,
\emph{Combinatorial Geometry},
John Wiley, New York, 1995.

\bibitem {Pa77} C. H. Papadimitriou,
Euclidean TSP is NP-complete,
\emph{Theoretical Computer Science} \textbf{4(3)} (1977), 237--244.

\bibitem {RS98}
S.~B. Rao and W.~D. Smith,
Approximating geometrical graphs via ``spanners'' and ``banyans,''
in {\em Proc. 30th ACM Symposium on Theory of Computing (STOC)},
ACM Press, 1998, pp.~540--550.

\bibitem{SS05}
S. Safra and O. Schwartz,
On the complexity of approximating TSP with neighborhoods
and related problems,
\emph{Computational Complexity} \textbf{14(4)} (2005), 281--307.

\bibitem{Sl97} P.~Slavik,
The errand scheduling problem,
CSE Technical Report 97-02, University of Buffalo, Buffalo, NY, 1997.

\bibitem{Sp14} S.~Spirkl,
The guillotine subdivision approach for TSP with neighborhoods revisited,
preprint, \url{arXiv:1312.0378v2}, 2014.

\bibitem{Va09}
K.~R. Varadarajan,
Epsilon nets and union complexity,
in \emph{Proc. 25th Symposium on Computational Geometry (SOCG)},
ACM Press, 2009, pp.~11--16.

\end{thebibliography}
\end{document}